\documentclass[reqno, 11pt]{amsart}
\usepackage[margin=3.5cm]{geometry}
\usepackage{amsmath, graphicx, amssymb, epstopdf, amsaddr}
\usepackage{cleveref}

\title{Exponential Dynamical Localization for Random Word Models}
\author{Nishant Rangamani}
\address{Department of Mathematics, University of California, Irvine}

\theoremstyle{plain}
\newtheorem{lem}{Lemma}
\newtheorem{thm}{Theorem}

\theoremstyle{definition}
\newtheorem*{defn}{Definition}

\theoremstyle{remark}
\newtheorem*{rmk}{Remark}

\begin{document}
\maketitle

\section{Abstract}

We show that one-dimensional Schr{\"o}dinger operators whose potentials arise by randomly concatenating words from an underlying set exhibit exponential dynamical localization (EDL) on any compact set which trivially intersects a finite set of critical energies. We do so by first giving a new proof of spectral localization for such operators and then showing that once one has the existence of a complete orthonormal basis of eigenfunctions (with probability one), the same estimates used to prove it naturally lead to a proof of the aforementioned EDL result. The EDL statements provide new localization results for several classes of random Schr{\"o}dinger operators including random polymer models and generalized Anderson models.

%We then show that once one has the existence of a complete orthonormal basis of eigenfunctions (with probability one), the same estimates used to prove it naturally lead to a proof of exponential dynamical localization in expectation (EDL) on any  compact set contained in the complement of the set of finite critical energies.

\section{Introduction}

In this paper, we consider random word models on $\ell^2(\mathbb{Z})$ given by  $$H_{\omega}\psi(n) = \psi(n+1) + \psi(n-1) +V_{\omega}(n)\psi(n).$$ The potential is a family of random variables defined on a probability space $\Omega$. To construct the potential $V$ above, we fix an $m\in \mathbb{N}$ (a maximum word length) then consider words $\ldots, \omega_{-1}, \omega_0, \omega_1, \ldots$ which are vectors in $\mathbb{R}^{n}$ with $1 \leq n \leq m$, so that $V_{\omega}(0)$ corresponds to the $k$th entry in $\omega_{0}$. A precise construction of the probability space $\Omega$ and the random variables $V_{\omega}(n)$ is carried out in \Cref{subsec:randomwordmodelsetup} and the precise definition of $H_{\omega}$ is given in \cref{eqn:wordop}.

These models are of particular interest not only because they cover a wide class of generalizations of the Anderson model such as the random dimer model, random polymer models, and generalized Anderson models \cite{damanikword}, but also because they provide natural examples of the subtleties involved in the various forms of localization: spectral, dynamical, and exponential dynamical (in expectation). 

Spectral localization occurs when almost surely, the spectrum is pure point and all of its eigenfunctions decay exponentially. A stronger form of localization, dynamical localization, is more closely related to physical manifestations of localization and implies an absence of transport in the random medium and that the wave packet is suitably confined. We say that $H_{\omega}$ exhibits dynamical localization on the interval $I$ if for a.e. $\omega$, and any $\psi \in \ell^{2}(\mathbb{Z})$ which decays exponentially, \begin{equation*}
\sup_{t}\langle P_I(H_\omega)e^{-iH_\omega t}\psi,|X|^qP_I(H_\omega)e^{-iH_\omega t}\psi\rangle < \infty,
\end{equation*}
where $P_I(H_\omega)$ is the spectral projection of $H_{\omega}$ onto the set $I$ and $q > 0$. 

While it was well known that dynamical localization implies spectral localization, that dynamical localization was a strictly stronger notion was not understood until the authors in  \cite{MR1428099} constructed an artificial model which was spectrally but not dynamically localized. Indeed, the example in \cite{MR1428099} showed that pure point spectrum with exponentially decaying eigenfunctions (spectral localization) could coexist with $\limsup_{t \to \infty} \frac{||xe^{-itH\delta_{0}||^{2}}}{t^{\alpha}} = \infty$ for all $\alpha < 2$. In short, spectral localization is not a sufficient condition to ensure an absence of transport, while the stronger notion of dynamical localization does imply this absence.

%The importance of dynamical results and their relevance to quantum transport is best understood through the appearance of $e^{-iH_\omega t}$ above, with the finitness of the overall expression implying an absence of transport in the corresponding random medium. 

A well-known physically relevant model which sheds more light on these phenomena is the random dimer model. This model was first introduced in \cite{dunlapwu} and in the random word context introduced above, $\omega_{i}$ takes values $(\lambda,\lambda)$ or $(-\lambda,-\lambda)$ with Bernoulli probability.
It is known that the spectrum of the operator $H_{\omega}$ is almost surely pure point with exponentially decaying eigenfunctions. On the other hand, when $0 < \lambda \leq 1$ (with $\lambda \neq \frac{1}{\sqrt{2}}$), there are critical energies at $E= \pm \lambda$ where the Lyapunov exponent vanishes \cite{germinetdimer}. These so-called critical energies are precisely what prevent the absence of transport and lead to localization-delocalization phenomena. 

In particular, the vanishing Lyapunov exponent at these energies can be exploited to prove lower bounds on quantum transport resulting in almost sure overdiffusive behavior \cite{stolz}. The authors in \cite{stolz} show that for almost every $\omega$ and for every $\alpha > 0$ there is a positive constant $C_{\alpha}$ such that 

\begin{equation*}
\frac{1}{T} \int_0^T \langle \delta_0,e^{iH_{\omega}t}|X|^qe^{-iH_{\omega}t}\delta_0 \rangle dt \geq C_\alpha T^{q-\frac{1}{2}-\alpha}.
\end{equation*}

This was later extended to a sharp estimate in \cite{sharp}.

%Such phenomena are of particular interest because while it is known that dynamical localization leads to an absence of transport and implies spectral localization, the converse is not true. For example, in \cite{MR1428099}, the authors construct a model which  displays spectral localization, but not dynamical localization.

In light of these remarks, the over-diffusive behavior above contrasts with the fact that not only does the random dimer model display spectral localization, but also dynamical localization  on any compact set $I$ not containing the critical energies $\pm \lambda$ \cite{germinetdimer}.

%\begin{equation}
%\sup_{t}\langle P_I(H_\omega)e^{-iH_\omega t}\psi,|X|^qP_I(H_\omega)e^{-iH_\omega t}\psi\rangle < \infty,
%\end{equation}
%Here $P_I(H_\omega)$ is the spectral projection of $H_{\omega}$ onto the set $I$.

We strengthen this last result by showing that there is  exponential dynamical localization in expectation (EDL) on any compact set $I$ with $\pm \lambda \notin I$. We say the family of operators $H_{\omega}$ display EDL on the interval $I$ if there are $C, \alpha > 0$ such that for any $p,q \in \mathbb{Z}$,
\begin{equation*}
\mathbb{E} \left[ \sup_{t \in \mathbb{R}} |\langle \delta_p, P_{I}(H_{\omega})e^{-itH_{\omega}}\delta_{q} \rangle|\right] \leq Ce^{-\alpha |p-q|}.
\end{equation*} 

EDL has several interesting physical consequences including exponential decay of the two point function in the ground state \cite{Aizenman_1998} and there is an interest in proving such results in physically relevant contexts such as the dimer and random polymer cases. Our results, however, when taken in conjunction with the over-diffusive behavior above illustrate that the strength of localization does not necessarily impact transport when the localization regime excludes critical energies.

One of the central challenges in dealing with random word models is the lack of regularity of the single-site distribution. The absence of regularity is exactly what allows random word models to encompass singular Anderson models, random dimer models, and more generally, random polymer models.   The issues presented by singularity were previously overcome using multi-scale analysis in various stages; first, in the Anderson setting \cite{ckm}, then in the dimer case \cite{germinetdimer,germinetdl}, and finally for random word models themselves in \cite{damanikword}. The multi-scale approach leads to weaker dynamical localization results than those where sufficient regularity of the single-site distribution allows one to instead appeal to the fractional moment method (e.g. \cite{elgartdl}, \cite{FMM}). In particular, EDL always follows in the framework of the fractional moment method and methodology in \cite{kunz1980spectre, Delyon1983OnedimensionalWE}, but of course  regularity is required.

Loosely speaking, the multi-scale analysis shows that the complement of the event where one has exponential decay of the Green's function has small probability. One of the consequences of this method is that while this event does have small probability, it can only be made sub-exponentially small. 

A recent new proof of spectral and dynamical localization for the one-dimensional Anderson model for arbitrary single-site distributions \cite{xj} uses positivity and large deviations of the Lyapunov exponent to replace parts of the multi-scale analysis. The major improvement in this regard (aside from a shortening of the length and complexity of localization proofs in one-dimension) is that the complement of the event where the Green's function decays exponentially can be shown to have exponentially (rather than sub-exponentially) small probability. These estimates were implicit in the proofs of spectral and dynamical localization given in \cite{xj} and were made explicit in \cite{GeZhao}. The authors in \cite{GeZhao} then used these estimates to prove EDL for the Anderson model and we extend those techniques to the random word case. 

There are, however, several issues one encounters when adapting the techniques developed for the Anderson model in \cite{GeZhao,xj} to the random word case. Firstly, in the Anderson setting, a uniform large deviation estimate is immediately available using a theorem in \cite{tsai}. Since random word models exhibit local correlations, there are additional steps that need to be taken in order to obtain suitable analogs of large deviation estimates used in \cite{GeZhao,xj}. Secondly, random word models may have a finite set of energies where the Lyapunov exponent vanishes and this phenomena demands care in obtaining estimates on the Green's functions analogous to those in \cite{GeZhao,xj}. Dealing with these issues does however, produce an unexpected benefit. Since we must consider Green's functions for non-symmetric intervals, we are able to obtain exponential decay of the Green's function centered around even and odd points simultaneously, while the arguments in \cite{GeZhao,xj} require separate considerations. 

%As above, with $P_{I}(H_{\omega})$ denoting the spectral projection of $H_{\omega}$ onto the interval $I$, we have the following result:

\begin{thm} \label{thm:specloc} With $H_\omega$ defined in \cref{eqn:wordop} (and satisfying \cref{eqn:NC}), for a.e. $\omega$, the spectrum of $H_{\omega}$ is pure point and all of its eigenfunctions decay exponentially. 
\end{thm}

Our main result is:

\begin{thm} \label{thm:EDL}
For $H_\omega$ defined in \cref{eqn:wordop} (and satisfying \cref{eqn:NC}), there is a finite $D \subset \mathbb{R}$ such that if $I$ is a compact set and $D \cap I = \emptyset$, then $H_{\omega}$ exhibits exponential dynamical localization in expectation in the interval $I$.

%there are $C>0$ and $\alpha > 0$ such that $\displaystyle \mathbb{E} [ \sup_{t \in \mathbb{R}} |\langle \delta_p, P_{I}(H_{\omega})e^{-itH_{\omega}}\delta_{q} \rangle|] \leq Ce^{-\alpha |p-q|}$ for any $p,q \in \mathbb{Z}$.
\end{thm}

The remainder of the paper is organized as follows: \cref{sec:prelim} contains preliminaries
needed to discuss the large deviation estimates found in \cref{sec:largedev}. \Cref{sec:lemmas}) contains the precursors needed for localization  and \cref{sec:specloc}
contains the proof of \Cref{thm:specloc} and \cref{sec:EDL} contains the proof of \Cref{thm:EDL}.

\section{Preliminaries} \label{sec:prelim}

\subsection{Model Set-up} \label{subsec:randomwordmodelsetup}

We begin by providing details on the construction of $\Omega$ and $V_{\omega}(n)$ by following \cite{damanikword}.

Fix $m \in \mathbb{N}$ (the maximum word length) and $M > 0$. Set $\mathcal{W} = \bigcup_{j=1}^m \mathcal{W}_{j}$ where $\mathcal W_{j} = [-M,M]^{j}$ and $\nu_{j}$ are finite Borel measures on $\mathcal{W}_{j}$ so that $\sum_{j=1}^{m} \nu_{j}(\mathcal{W}_{j}) = 1$. Let $\nu$ denote the direct sum of the measures $\nu_j$, a probability measure on $\mathcal{W}$.

Additionally, we assume that $(\mathcal{W},\nu)$ has two words which do not commute. That is, \begin{equation} \label{eqn:NC}
\left\{ \begin{array}{lc}
\text{For} \hspace{0.1cm} i = 0,1 \hspace{0.1cm} \text {there} \hspace{0.1cm} \text{exist} \hspace{0.1cm} w_{i} \in \mathcal{W}_{j_{i}} \in supp(\nu) \hspace{0.1cm} \hspace{0.1cm} \text{such} \hspace{0.1cm} \text{that}\\
(w_{0}(1),w_{0}(2),...,w_{0}(j_{0}),w_{1}(1),w_{1}(2),...,w_{1}(j_{1}))\\  
\text{and} \hspace{0.1cm} (w_{1}(1),w_{1}(2),...,w_{1}(j_{1}),w_{0}(1),w_{0}(2),...,w_{0}(j_{0})) \hspace{0.1cm} \text{are} \hspace{0.1cm} \text{distinct.} \\
\end{array}
\right.
\end{equation} Here, $supp(\nu)$ refers to the support of the measure $\nu$. 

Set $\Omega_{0} = \mathcal{W}^{\mathbb{Z}}$ and $\mathbb{P}_{0} = \otimes_{\mathbb{Z}} \nu$ on the $\sigma$-algebra generated by the cylinder sets in $\Omega_{0}$. 

The average length of a word is defined by $\langle L \rangle = \sum_{j=1}^{m} j \nu(\mathcal{W}_{j})$ and if $w \in \mathcal{W} \cap \mathcal{W}_{j}$, we say $w$ has length $j$ and write $|w| = j$. 

We define $\Omega = \bigcup_{j=1}^{m} \Omega_{j} \subset \Omega_{0} \times \{1,...,m\}$ where $\Omega_{j} = \{\omega \in \Omega_{0} : |\omega_{0}| = j\} \times \{1,...,j\}$. We define the probability measure $\mathbb{P}$ on the $\sigma$-algebra generated by the sets $A \times \{k\}$ where $A \subset \Omega_{0}$ such that for all $\omega \in A$, $|\omega_{0}|=j$ and $1 \leq k \leq j$. 

For such sets we set 
\begin{equation}\label{eqn:lifting}\mathbb{P}[A \times \{k\}] = \frac{\mathbb{P}_{0}(A)}{\langle L \rangle}.
\end{equation}

\begin{rmk} The above construction implies that every event $A \subset \Omega_{0}$ gives rise to an event $\tilde{A} \subset \Omega$ with the same probability (up to multiplication by $\langle L \rangle$).

\end{rmk}

The shifts $T_{0}$ and $T$ on $\Omega_{0}$ and $\Omega$ (respectively) are given by:

$(T_{0}\omega)_n = \omega_{n+1}$ and  \begin{equation*}
T(\omega,k)=
\left\{ \begin{array}{lc}
(\omega, k+1) & \text{ if } k< |\omega_{0}| \\
(T_0(\omega),1) & \text { if } k = |\omega_{0}|. \\
\end{array}
\right.
\end{equation*}

With this set-up, the shift $T$ is ergodic and the potential $V_{\omega,k}$ is obtained through $...,\omega_{-1},\omega_{0},\omega_{1},...$ so that $V_{\omega,k}(0) = \omega_{0}(k)$.    

That is, \begin{equation} \label{eqn:wordop}
(H_{(\omega,k)}u)(n) = u(n+1) + u(n-1) + V_{(\omega,k)}(n)u(n)
\end{equation} for all $u \in \ell^{2}(\mathbb{Z})$.

Thus, the ergodicity of the shift $T$ implies the results from \cite{cmp/1103908097} can be used to show the spectrum of $H_{(\omega,k)}$ is almost surely a non-random set. 

\begin{rmk} For notational convenience, we will often drop the $k$ from the subscript on $H_{(\omega,k)}$.

\end{rmk}

\subsection{Basic Definitions and Notations}

\begin{defn} We call $\psi_{\omega,E}$ a \textit{generalized eigenfunction} with \textit{generalized eigenvalue} $E$ if $H_{\omega}\psi_{\omega,E}=E\psi_{\omega,E}$ and $|\psi_{\omega,E}(n)| \leq (1+|n|)$.

\end{defn}

We denote the restriction of $H_{\omega}$ to the interval $[a,b] \cap \mathbb{Z}$ where $a,b \in \mathbb{Z}$ by $H_{\omega,[a,b]}$ and for $E\notin \sigma(H_{\omega,[a,b]})$ the corresponding Green's function by $$G_{[a,b],E,\omega}=(H_{\omega,[a,b]}-E)^{-1}.$$
Additionally, we let 
$$P_{[a,b],E,\omega}=\det(H_{\omega,[a,b]}-E)$$ and 
$$\tilde{P}_{[a,b],E,\omega,}=\det(E-H_{\omega,[a,b]}).$$

We also let $E_{j,[a,b],\omega}$ denote the $j$th eigenvalue of the operator $H_{\omega,[a,b]}$ and note that there are $b-a+1$ many of them (counting multiplicity). 

\begin{defn}
$x\in \mathbb{Z}$ is called $(c,n_1,n_2,E,\omega)$-\textit{regular} if there is a $c>0$ so that:
\begin{enumerate}
\item $|G_{[x-n_{1},x+n_{2}],E,\omega}(x,x-n_{1})|\leq e^{-cn_{1}}$ and
\item $|G_{[x-n_{1},x+n_{2}],E,\omega}(x,x+n_{2})|\leq e^{-cn_{2}}$.
\end{enumerate}
\end{defn}

By a standard decoupling argument, we have for any generalized eigenfunction $\psi_{\omega,E}$ and any $x \in [a,b]$,

\begin{equation} \label{eqn:greenef}
 \psi_{\omega,E}(x)=-G_{[a,b],E,\omega}(x,a)\psi_{\omega,E}(a-1)-G_{[a,b],E,\omega}(x,b)\psi_{\omega,E}(b+1), 
\end{equation}

and 

\begin{equation} \label{eqn:greendet}
 |G_{[a,b],E,\omega}(x,y)|=\frac{|P_{[a,x-1],E,\omega}P_{[y+1,b],E,\omega}|}{|P_{[a,b],E,\omega}|}.
\end{equation}

\subsection{Transfer Matrices and the Lyapunov Exponent}

For $w \in \mathcal{W}$, with $w = (w_{1},...,w_{j})$, we define word transfer matrices by $T_{w,E} = T_{w_{j},E} \cdots T_{w_{1},E}$ where $T_{v,E} = \begin{pmatrix}
E-v & -1\\
1   & 0
\end{pmatrix}$.

The transfer matrices over several words are given by 
\begin{equation} \label{eqn:wordtransfer}
T_{\omega,E}(k,l) =
\left\{ \begin{array}{lc}
T_{\omega_{k},E} \cdots T_{\omega_{l},E} & \text{ if } k > l, \\
\mathbb{I} & \text{ if } k = l,\\
{T_{\omega,E}}^{-1}(l,k) & \text { if } k < l. \\
\end{array}
\right.
\end{equation}

and $T_{[a,b],E,\omega}$ denotes the product of the transfer matrices so that for any generalized eigenfunction $\psi$ with generalized eigenvalue $E$:
\begin{equation*}
T_{[a,b],E,\omega}\begin{pmatrix}
\psi(a)\\
\psi(a-1)\\
\end{pmatrix}=\begin{pmatrix}
\psi(b+1)\\
\psi(b)\\
\end{pmatrix}.
\end{equation*}

An inductive argument shows:

\begin{align} \label{eqn:detrest}
T_{[a,b],E,\omega}=\begin{pmatrix}
\tilde{P}_{[a,b],E,\omega} & -\tilde{P}_{[a+1,b],E,\omega}\\
\tilde{P}_{[a,b-1],E,\omega} & -\tilde{P}_{[a+1,b-1],E,\omega}\\
\end{pmatrix}.
\end{align}

We now define two Lyapunov exponents, one via matrix products obtained from $\Omega_{0}$ and the other from $\Omega$. We will see that the two quantities are essentially the same (up to multiplication by a positive constant) and hence provide the same information. It is worth noting, however, that the former is obtained through products of i.i.d. matrices while the latter is not. In fact, we will need to utilize independence to obtain various estimates on the matrix products and it is therefore important to verify the relationship between the two Lyapunov exponents.

Since both $T_{0}$ and $T$ are ergodic, Kingman's subadditive ergodic theorem \cite{10.1214/aop/1176996798} allows us to define the Lyapunov exponent. Recalling that $\langle L \rangle$ denotes the average length of a word, by the arguments in \cite{damanikword}, we have:

in $\Omega_{0}$,
\begin{equation*}
\gamma_{0}(E) := \lim_{k \to \infty} \frac{1}{k} \log || T_{\omega,E}(k,1)||,
\end{equation*}

and in $\Omega$,

\begin{equation*}
\gamma(E) := \lim_{k \to \infty} \frac{1}{k\langle L \rangle} \log \left|\left| \begin{pmatrix}
E-V_{\omega}(n) & -1\\
1   & 0
\end{pmatrix} \hdots \begin{pmatrix}
E-V_{\omega}(1) & -1\\
1   & 0
\end{pmatrix}\right|\right|.
\end{equation*}

In both cases, the limit exists for fixed $E$ on a full measure set. 

\begin{rmk} We note that the limit in $\Omega$ is defined via one-step transfer matrices while the limit in $\Omega_{0}$ is defined via word transfer matrices.

\end{rmk}

In \cite{damanikword}, the authors prove the relationship between the two Lyapunov exponents described in the following theorem.

\begin{thm} $\frac{\gamma_{0}(E)}{\langle L \rangle} = \gamma(E)$.

\end{thm}

Let $\mu_E$ denote the smallest closed subgroup of $SL(2,\mathbb{R})$ generated by the `word'-step transfer matrices. It is shown in \cite{damanikword} that $\mu_E$ is strongly irreducible and contracting for all $E$ outside of a finite set $D \subset \mathbb{R}$ and hence, Furstenberg's theorem implies $\gamma(E) >0$ for all such $E$. Since the Lyapunov exponent is defined as a product of i.i.d. matrices,  $\gamma$ is continuous. So, if $I$ is a compact set such that $D \cap I = \emptyset$ and 
$$\nu:=\inf\{\gamma(E):E \in I\}, $$
then $\nu>0$.

Motivated by \cref{eqn:detrest} above and large deviation theorems, we define:
 
\begin{equation*}
 {B^{+}}_{[a,b],\varepsilon}=\left\{(E,\omega): E\in I,{|P_{[a,b],E,\omega}|}\geq e^{(\gamma(E)+\varepsilon)(b-a+1)}\right\},
 \end{equation*} 
 and
 
\begin{equation*}
 {B^{-}}_{[a,b],\varepsilon}=\left\{(E,\omega): E \in I,{|P_{[a,b],E,\omega}|}\leq e^{(\gamma(E)-\varepsilon)(b-a+1)}\right\},
\end{equation*}

and the corresponding sections:

\begin{equation*}
{B^{\pm}}_{[a,b],\varepsilon,\omega}=\left\{E: (E,\omega) \in B^{\pm}_ {[a,b],\varepsilon}\right\},
\end{equation*}
and 
\begin{equation*}
 {B^{\pm}}_{[a,b],\varepsilon,E}=\left\{\omega: (E,\omega) \in B^{\pm}_ {[a,b],\varepsilon}\right\}.
\end{equation*}

\section{Large Deviation Theorems} \label{sec:largedev}

The goal of this section is to obtain a uniform large deviation estimate for $P_{[a,b],E}$. In the Anderson model, a direct application of Tsay's theorem for matrix elements of products of i.i.d. matrices results in both an upper and lower bound for the above determinants. In the general random word case, there are two issues. Firstly, the one-step transfer matrices are not independent. This issue is naturally resolved by considering $\omega_{k}$-step transfer matrices and treating products over each word as a single step. However, in this case, both the randomness in the length of the chain as well as products involving partial words need to be accounted for. Since matrix elements are majorized by the norm of the matrix and all matrices in question are uniformly bounded, we can obtain an upper bound identical to the one obtained in the Anderson case. Lower bounds on the matrix elements are more delicate and require the introduction of random scales. For the reader's convenience, we first recall Tsay's theorem and then give the precise statements and proofs of the results alluded to above.

As remarked above, $\mu_E$ is strongly irreducible and contracting for $E \in I$. In addition, the `word'-step transfer matrices (defined in \cref{eqn:wordtransfer}) are bounded, independent, and identically distributed. These conditions are sufficient to apply Tsay's theorem on large deviations of matrix elements for products of i.i.d. matrices. 

\begin{thm}[\cite{tsai}] \label{thm:LDT} Suppose $I$ is a compact interval and for each $E \in I$, $Z_{1}^{E},...,Z_{n}^{E},...$ are bounded i.i.d random matrices such that the smallest closed subgroup of $SL(2,\mathbb{R})$ generated by the matrices is strongly irreducible and contracting. Then for any $\varepsilon > 0$, there is an $\eta > 0$ and an $N\in\mathbb{N}$ such that for any $E \in I$, for any unit vectors $u$, $v$ and $n > N$, $$\mathbb{P}\left[ e^{(\gamma(E) - \varepsilon)n} \leq | \langle Z_{n}^{E} \hdots Z_{1}^{E}u,v \rangle| \leq e^{(\gamma(E) + \varepsilon)n}\right] \geq 1 - e^{-\eta n}.$$
 
\end{thm} 

\begin{lem} \label{thm:lrgdev} If $I$ is compact and $I \cap D = \emptyset$, then for any $\varepsilon > 0$ there is an $\eta > 0$ and an $N$ such that if $a$,$b \in \mathbb{Z}$ such that $b-a+1 > N$ and  $E \in I$, then
$$\mathbb{P}[B^{+}_{[a,b],\varepsilon,E}] \leq e^{-\eta(b-a+1)}.$$
\end{lem}

\begin{proof} Let $Y_i = |\omega_{i}|$, so $Y_i$ is the length of the $i$th word and let $S_n = Y_1 +\cdots + Y_n$.

Let $u = \begin{pmatrix}
\ 1 \\
\ 0
\end{pmatrix}$ %$v = \begin{pmatrix}
%\ 0 \\
%\ 1
%\end{pmatrix}$
and let $P_{(\omega_{1},\omega_{n}),E} = \det(H_{(\omega_{1},\omega_{n})}-E)$ where $H_{(\omega_{1},\omega_{n})}$ denotes $H_{\omega}$ restricted to the interval where $V$ takes values determined by $\omega_{1}$ through $\omega_{n}$. By \cref{eqn:detrest} from the previous section, $|P_{(\omega_{1},\omega_{n}),E}| = |\langle T_{\omega,E}(n,1) u,u \rangle|$.

Letting $\varepsilon > 0$ and applying \Cref{thm:LDT} to the random products $T_{\omega,E}(k,1)$, we obtain an $\eta_{1} > 0$ and an $N_1$ such that for $n > N_{1}$, $\mathbb{P}_{0}[\{\omega \in \Omega_{0}: |P_{(\omega_{1},\omega_{n}),E}| \leq e^{(\gamma(E) + \varepsilon) n\langle L \rangle}\}] \geq 1- e^{-\eta_{1}n}$.

Now let $\varepsilon_{1} > 0$ so that $\varepsilon_{1}\sup\{\gamma(E): E \in I\} < \varepsilon$. We apply large deviation estimates (e.g. \cite{durrett}) to the real, bounded, i.i.d. random variables $Y_i$ to obtain an $N_{2}$ and an $\eta_{2} > 0$ such that for $n > N_{2}$, $\mathbb{P}_{0} [S_n - n\varepsilon_{1} < n\langle L \rangle < S_n + n\varepsilon_{1}] \geq 1 - e^{-\eta_{2}n}$.

Denoting the intersection of the above events by $A_n$, we have, on $A_{n}$,

\begin{align*}
|P_{(\omega_{1},\omega_{n}),E}| &\leq e^{(\gamma(E) + \varepsilon)n \langle L \rangle } \\
&\leq e^{(\gamma(E) + \varepsilon)(S_n + n\varepsilon_{1})} \\
&= e^{(\gamma(E) + \varepsilon)S_n +\gamma(E)\varepsilon_{1}n + n\varepsilon\varepsilon_{1}} \\
& \leq e^{(\gamma(E) + \varepsilon)S_n +\gamma(E)\varepsilon_{1}S_n + S_n\varepsilon\varepsilon_{1}} \\
& \leq e^{(\gamma(E) + 3 \varepsilon)S_n}.
\end{align*}

Thus, we have an event in $\Omega_{0}$ where the Lyapunov behavior is a true reflection of the length of the interval and we can obtain an estimate in between two words.

That is, for any $1 \leq k \leq S_{n+1} - S_{n}$, let $P_{(\omega_{1},\omega_{n}+k),E} = \det(H_{(\omega_{1},\omega_{n} + k)}-E)$ where $H_{(\omega_{1},\omega_{n} + k)}$ denotes $H_{\omega}$ restricted to the interval where $V_{\omega}$ takes values determined by $\omega_{1}$ through the $k$th letter of $\omega_{n+1}$.

Since the one-step transfer matrices are uniformly bounded, \cref{eqn:detrest} and the last inequality imply for any $1 \leq k \leq S_{n+1} - S_{n}$, $|P_{[1,S_{n}+k],E,\omega}| \leq Ce^{(\gamma(E) + 3 \varepsilon)S_n} \leq e^{(\gamma(E) + 4 \varepsilon)S_n}$  on $A_n$.

Let $\eta_{3} = \min\{\eta_{1},\eta_{2}\}$, and choose $0<\eta<\frac{\eta_{3}}{2m}$. 
%\begin{align*}
%\mathbb{P}[A_n]& \geq 1 - e^{-2\eta_{3}n} \\
%&\geq 1 - e^{-\eta_{4}(S_{n} + k)}.
%\end{align*}

Since every event in $\Omega_{0}$ gives rise to an event in $\Omega$ (e.g. \cref{eqn:lifting}), we obtain an estimate where the Lyapunov behavior and the probability of the event reflect the true length of the interval, so we can apply the shift $T$ to conclude that for any sufficiently large $n$, 
$$\mathbb{P}[\{|P_{[1,n],E,\omega}| \leq e^{(\gamma(E) + 4 \varepsilon)n}\}] \geq 1 - e^{-\eta n}.$$

The result now follows for any interval $[a,b]$ (with $b-a +1$ sufficiently large) since $T$ preserves the probability of events.

\end{proof}

We finish the section with a lemma that relies crucially on independence. As above, we will work in $\Omega_{0}$ and `lift' our results to $\Omega$.

The lemma below holds for any fixed $K>1$ and this $K$ will be chosen in the next section.

\begin{lem}\label{thm:boxev2} There are real-valued random variables $R_{n}$, $R^{'}_{n}$, $Q_{n}$, $Q^{'}_{n}$, and $\tilde{Q}_{n}$ such that:
%\begin{enumerate}
%\item $ \frac{n(\langle L \rangle - \varepsilon)+1}{2}   \leq R_n \leq \frac{n(\langle L \rangle + \varepsilon) + m}{2} $,

%\item $R_{n}+1 \leq Q_{n} \leq R_{n} + m$,

%\item $(n+1)(\langle L \rangle - \varepsilon)   \leq Q^{'}_{n} \leq (n+1)(\langle L \rangle + \varepsilon)$,
%\end{enumerate}
if $I$ is compact with $I \cap D = \emptyset$ and $\varepsilon > 0$ with

\begin{equation*}
F_{l,n,\varepsilon}^{3} = \bigcap_{j} \left\{\omega: E_{j, [l+Q_{n}+k,l+Q^{'}_{n}],\omega} \notin B^{-}_{[l+R_{n},l+R^{'}_{n}],\varepsilon,\omega} \;\forall k, \hspace{0.1cm} 0 \leq k \leq 2m \right\},
\end{equation*}

\begin{equation*}
F_{l,n,\varepsilon}^{2+} = \bigcap_{j} \left\{\omega: E_{j,[l+Q_{n}+k,l+Q^{'}_{n}],\omega} \notin B^{+}_{[y,l+R^{'}_{n}],\varepsilon,\omega} \;\forall y \in \left[l+R_{n}, l+R^{'}_{n} - \frac{n}{K}\right], 0 \leq k \leq 2m\right\},
\end{equation*}

\begin{equation*}
F_{l,n,\varepsilon}^{2-} = \bigcap_{j} \left\{\omega: E_{j,[l+Q_{n}+k,l+Q^{'}_{n}],\omega} \notin B^{+}_{[l+R_{n},y],\varepsilon,\omega} \forall y \in \left[l+R_{n} +\frac{n}{K}, l+R^{'}_{n}\right], 0 \leq k \leq 2m \right\},
\end{equation*}
and
\begin{equation*}
F_{l,n,\varepsilon}^{2} = F_{l,n,\varepsilon}^{2+} \cap F_{l,n,\varepsilon}^{2-},
\end{equation*}

 %For any $\varepsilon > 0$ there is an $\eta>0$, sets $A_{n,l} \subset \Omega$, a sequence of random scales $R_n$ and $Q_n$, and an $N$ such that for $n>N$ and every $l \in \mathbb{Z}$, on $A_{n,l}$

%\begin{enumerate}

%\item $S_n - n\varepsilon_{1} < n\langle L \rangle < S_n + n\varepsilon_{1}$,

%\item $ \frac{n(\langle L \rangle - \varepsilon)}{2} - m \leq R_n \leq \frac{n(\langle L \rangle + \varepsilon)}{2} -m$,

%\item $ n(\langle L \rangle - \varepsilon) +2m  \leq Q_n \leq n(\langle L \rangle + \varepsilon) + 2m$,

%\item $|P_{\omega,[-R_{n},R_{n}],E}| \geq e^{(\gamma(E) - \varepsilon)2R_{n}}$,

%\item for any $S_{n} + 2m \leq k < S_{n+1} + 2m$, $|P_{\omega,[k-R_{n},k+R_{n}],E}| \geq e^{(\gamma(E) - \varepsilon)2R_{n}}$,

%\item $\mathbb{P}[A_{n,l}] \geq 1 - e^{-\eta 2R_{n}}$.

%\end{enumerate}

then there is $N$ and an $\eta' > 0$ such that if $n>N$, $l \in \mathbb{Z}$, and $E \in I$:

\begin{equation*}\mathbb{P}[B^{-}_{[l+R_{n},l+R^{'}_{n}],\varepsilon,E}] \leq e^{-\eta'(2n+1)},
\end{equation*}

\begin{equation*}\mathbb{P}[B^{-}_{[l+Q_{n},l+Q^{'}_{n}],\varepsilon,E}] \leq e^{-\eta'(2n+1)},
\end{equation*}

\begin{equation*}\mathbb{P}[F_{l,n,\varepsilon}^{3}] \geq 1 - 2m^{2}(2n+3)^{2}e^{-\eta'(2n+1)},
\end{equation*}

and
\begin{equation*}\mathbb{P}[F_{l,n,\varepsilon}^{2}] \geq 1 - 2m^{4}(2n+3)^{3}e^{-\eta'(\frac{n}{K})}.
\end{equation*}
\end{lem}

\begin{proof}

%Set $\tilde{F}_{l,n,\varepsilon}^{3} = \bigcap_{j} 
%\{\omega: E_{j,\omega, [n+1,3n+1]} 
%\notin B^{-}_{[-n,n],\varepsilon,\omega}\}$. 

For $\omega \in \Omega_0$, and $a,b \in \mathbb{Z}$, let $H_{(\omega_{a},\omega_{b})}$ denote the restriction of $H$ to the interval where the potential is given by the words $\omega_{a}$ through $\omega_{b}$. Take $u = \begin{pmatrix}
\ 1 \\
\ 0
\end{pmatrix}$, 
%$v = \begin{pmatrix}
%\ 0 \\
%\ 1
%\end{pmatrix} 
and let $P_{(\omega_{a},\omega_{b}),E} = \det(H_{(\omega_{a},\omega_{b})}-E)$ where $H_{(\omega_{a},\omega_{b})}$ denotes $H_{\omega}$ restricted to the interval in which $V$ takes values determined by $\omega_{a}$ through $\omega_{b}$. By \cref{eqn:detrest} from the previous section, $|P_{(\omega_{-n},\omega_{n}),E}| = |\langle T_{\omega,E}(n,-n) u,u \rangle|$. Finally, let $S_{(a,b)} = |\omega_{a}|+ \cdots + |\omega_{b}|$. Letting $\varepsilon > 0$ and applying \Cref{thm:LDT} to the random products $T_{\omega,E}(k,1)$, we obtain an $\eta_{1} > 0$ and an $N_1$ such that for $n > N_{1}$ and any $E \in I$, 
\begin{equation}
\mathbb{P}_{0}[\{\omega \in \Omega_{0} : |P_{(\omega_{-n},\omega_{n}),E}| \geq e^{(\gamma(E) - \varepsilon) (2n+1)\langle L \rangle}\}] \geq 1- e^{-\eta_{1}(2n+1)}.
\end{equation} \label{eqn:4.24}
and 
\begin{equation} \label{eqn:4.25}
\mathbb{P}_{0}[\{\omega \in \Omega_{0} : |P_{(\omega_{n+1},\omega_{3n+1}),E}| \geq e^{(\gamma(E) - \varepsilon) (2n+1)\langle L \rangle}\}] \geq 1- e^{-\eta_{1}(2n+1)}.
\end{equation}

If $E_{j}$ denotes an eigenvalue corresponding to $H_{(\omega_{n+1},\omega_{3n+1})}$, then by independence $$\mathbb{P}_{0}[\{\omega \in \Omega_{0} : |P_{(\omega_{-n},\omega_{n}),E_{j}}| \geq e^{(\gamma(E_{j}) + \varepsilon)(2n+1)\langle L \rangle}\}] \geq 1- e^{-\eta_{1}(2n+1)}$$ whenever $n > N$. We can also apply large deviation theorems (e.g. \cite{durrett}) to the (real) random products $S_{(a,b)}$ where $0 <\varepsilon_{1} < \min\{1,\varepsilon,\varepsilon\sup\{\gamma(E): E \in I\}\}$ to obtain an $N$ and an $\eta_{2} > 0$  such that whenever $b-a +1 > N$,
\begin{equation} \label{eqn:4.26}
(b-a+1)(\langle L \rangle - \varepsilon_{1}) \leq S_{(a,b)} \leq (b-a+1)(\langle L \rangle + \varepsilon_{1})
\end{equation} 
with probability greater than $1-e^{-\eta_{2}(b-a+1)}$.

%Moreover, by choosing $0 < \eta_{2} < \tilde{\eta_{2}}$ and increasing $N$, we can ensure the same estimate holds for any $a,b \in [-n,n] or [n+1,3n+1]$ whenever $b-a > N$ with probability greater than $1 - e^{-\eta(b-a+1)}$.

Thus, we have an event $A$ where:

\begin{align} \label{eqn:4.27}
\begin{split}
|P_{(\omega_{-n},\omega_{n}),E_j}| &\geq e^{(\gamma(E) - \varepsilon)(2n+1) \langle L \rangle } \\
&\geq e^{(\gamma(E) - \varepsilon)(S_{(-n,n)} - (2n+1)\varepsilon_{1})} \\
&= e^{(\gamma(E) - \varepsilon)S_{(-n,n)} -(\gamma(E)\varepsilon_{1})(2n+1) + (2n+1)\varepsilon\varepsilon_{1}}\\
& \geq e^{(\gamma(E) - \varepsilon)S_{(-n,n)} -(\gamma(E)\varepsilon_{1})S_{(-n,n)} + (2n+1)\varepsilon\varepsilon_{1}}\\
& \geq e^{(\gamma(E)  - 2\varepsilon)S_n}.
\end{split}
\end{align}

By using a similar argument to deal with the upper-bound, we have:

$|P_{(\omega_{-n},\omega_{n}),E_j}|\leq e^{(\gamma(E)  + 3\varepsilon)S_{(-n,n)}}$.

In particular, whenever $\frac{n}{K} > N$, and $y \in [-n + \frac{n}{K},n]$, we have:
\begin{equation} \label{eqn:4.28}
|P_{E_{j},(\omega_{y},\omega_{n})}|\leq e^{(\gamma(E)  + 3\varepsilon)S_{(-n+\frac{n}{K},n)}}
\end{equation}
 with probability greater than $1 - e^{\eta_{1}\frac{n}{K}}$ with a similar estimate holding whenever $y \in [-n,n-\frac{n}{K}]$.

Since the single-step transfer matrices are uniformly bounded and there are at most $m$ single-step transfer matrices in a word transfer matrix, we have a $C>0$ such that:

\begin{equation*}\mathbb{P}_{0}[\{ \omega \in \Omega_{0} : |P_{(\omega_{y}+k,\omega_{n}),E_j}| \leq Ce^{(\gamma(E_{j}) + \varepsilon)(n-y)}\}] \leq e^{-\eta_{1}(n-y-1)},
\end{equation*}

where $P_{(\omega_{y}+k,\omega_{n})}$ denotes the determinant obtained by restricting $H$ from the $k$th letter of $\omega_{y}$ to $\omega_{n}$.
 
%Again, by the remark after \cref{eqn:lifting}, obtain an event $A_{1}$ in $\Omega$ where the estimates above hold. 

%Now for all $1 \leq j \leq m$, we apply the shift operator $T$ for all $k$, $\frac{n(\langle L \rangle - \varepsilon) + j}{2} \leq k \leq \frac{n(\langle L \rangle + \varepsilon) + j}{2}$ and take the intersection of these events to obtain random variables $R_{n}$, $Q_{n}$, and $Q^{'}_{n}$ such that:

%\begin{enumerate}

%\item $ \frac{n(\langle L \rangle - \varepsilon)+1}{2}   \leq R_n \leq \frac{n(\langle L \rangle + \varepsilon) + m}{2} $,

%\item $R_{n}+1 \leq Q_{n} \leq R_{n} + m$

%\item $(n+1)(\langle L \rangle - \varepsilon)   \leq Q^{'}_{n} \leq (n+1)(\langle L \rangle + \varepsilon)$

%\item $\mathbb{P}[\{\omega: E_{j,\omega, [Q_{n}+k,Q^{'}_{n}]} \notin B^{-}_{[-R_{n},R_{n}],\varepsilon,\omega} \hspace{0.1cm} \text{for all} \hspace{0.1cm} k \hspace{0.1cm} \text{with} \hspace{0.1cm} 0 \leq k < Q_{n+1} - Q_{n} \}] \leq e^{-\eta(2R_{n}+1)}$

%\end{enumerate}

Note that the products $S_{(a,b)}$ are also well-defined on $\Omega$ and we now are able to define the random variables from the statement of the lemma.

For $(\omega,k) \in \Omega$ we put,

\begin{enumerate}

\item $R^{'}_{n} = S_{(0,n)} - k$,

\item $R_{n} = -S_{(-n,-1)} - k+1$,

\item $Q_{n} = R^{'}_{n} + 1$,

\item $Q^{'}_{n} = R^{'}_{n} + S_{(n+1,3n+1)}$, and

\item  $\tilde{Q}_{n} = R^{'}_{n} + S_{(n+1,2n+1)}$.

\end{enumerate}

Choosing $0 < \tilde{\eta_{2}} < \eta_{2}$ and applying \cref{eqn:4.26}, we obtain an event with probability greater than $1 - e^{-\tilde{\eta_{2}}n}$ where we have estimates on all of the random variables defined above.

We choose $\eta'> 0$ to be smaller than $\eta_{1}$ and $\tilde{\eta_{2}}$ so that by the remark after \cref{eqn:lifting}, the definitions of the random variables above, and \cref{eqn:4.24} and \cref{eqn:4.25}, we have for any $E \in I$:

 $$\mathbb{P}[B^{-}_{[l+R_{n},l+R^{'}_{n}],E}] \leq e^{-\eta'(2n+1)}$$
and

$$\mathbb{P}[B^{-}_{[l+Q_{n},l+Q^{'}_{n}],E}] \leq e^{-\eta'(2n+1)}.$$

Moreover, by the same reasoning, after taking the union over all of the possible eigenvalues $E_{j}$ and the ordered pairs $(y,n)$ (and $(-n,y)$), \cref{eqn:4.27} and \cref{eqn:4.28} provide the desired estimates on 

$$F_{0,n,\varepsilon}^{3} = \bigcap_{j} \left\{\omega: E_{j,[Q_{n}+k,Q^{'}_{n}],\omega} \notin B^{-}_{[R_{n},R^{'}_{n}],\varepsilon,\omega} \hspace{0.1cm} \text{for all} \hspace{0.1cm} k \hspace{0.1cm} \text{with} \hspace{0.1cm} 0 \leq k \leq 2m \right\}$$

and $F_{0,n,\varepsilon}^{2}.$

The result now follows by applying the (measure-preserving) shift $T$ so that the intervals are centered around $l$ rather than $0$.

\end{proof}

\begin{rmk} Each of the lemmas above furnish a positive constant: $\eta$, $\eta^{'}$. For a fixed $\varepsilon$, we call the the minimum of these constants the `large deviation parameter' associated with $\varepsilon$ and denote it by $\eta_{\varepsilon}$.

\end{rmk}

\section{Lemmas}\label{sec:lemmas}

We prove localization results on a compact interval $I$ where $D \cap I = \emptyset$. In order to do so, we fix a larger interval $\tilde{I}$ such that $I$ is properly contained in $\tilde{I}$ and $D \cap \tilde{I} = \emptyset$, then apply the large deviation theorems from the previous section to $\tilde{I}$.

The following lemmas involve parameters $\varepsilon_0, \varepsilon, \eta_0, \delta_0, \eta_\varepsilon$,$K$, and the intervals $I,\tilde{I}$. The lemmas hold for any values satisfying the constraints below: 
\begin{enumerate}
\item Let $\nu =\inf \{\gamma(E): E \in \tilde{I}\}$, take $0<\varepsilon_0<\nu/8$ and let $\eta_0$ denote the large deviation parameter corresponding to $\varepsilon_0$. Choose any $0<\delta_0<\eta_0$ and let $0<\varepsilon<\min\{(\eta_0-\delta_0)/3m,\varepsilon_{0}/4\}$. Choose $\tilde{M}>0$  so that $|P_{[a,b],E,\omega}|\leq \tilde{M}^{b-a+1}$ for all intervals $[a,b]$, $E\in \tilde{I}$, and $\omega \in \Omega.$
Lastly, choose $K$ so that $\tilde{M}^{1/K}<e^{\nu/2}$ and let $\eta_\varepsilon$, $\eta_{\frac{\varepsilon}{4}}$ denote the large deviation parameters corresponding to $\varepsilon$ and $\frac{\varepsilon}{4}$ respectively. 
\item Any $N$'s and constants furnished by the lemmas below depend only on the parameters above (i.e. they are independent of $l \in \mathbb{Z}$ and $\omega$).

\end{enumerate}
 Thus, for the remainder of the work, $\varepsilon_0, \varepsilon, \eta_0, \delta_0, \eta_\varepsilon,$ and $K$ will be treated as fixed parameters chosen in the manner outlined above. 
 
\begin{rmk} Note that the sets $B^{\pm}_{[a,b],\varepsilon}$ are hereafter defined in terms of $\tilde{I}$ rather than $I$. 
\end{rmk}

Following \cite{xj} and \cite{GeZhao}, we define subsets of $\Omega$ below on which we have regularity of the Green's functions. This is the key to the proof of all the localization results. As mentioned in the introduction, the proofs of spectral and dynamical localization given in \cite{xj} show that an event formed by the complement of the sets below has exponentially small probability. These estimates were exploited in \cite{GeZhao} to provide a proof of exponential dynamical localization for the one-dimensional Anderson model. We follow the example set in these two papers with appropriate modifications needed to handle the presence of critical energies and the varying length of words.

Let $m_{L}$ denote Lebesgue measure on $\mathbb{R}$.

\begin{lem} \label{thm:singpts}
If $n\geq 2$ and $x$ is $(\gamma(E)-8\varepsilon_0,n_{1},n_{2},E,\omega)$-singular, then $$(E,\omega)\in {B^-}_{[x-n_{1},x+n_{2}],\varepsilon_0}\cup {B^+}_{[x-n_{1},x-1],\varepsilon_0}\cup {B^+}_{[x+1,x+n_{2}],\varepsilon_0}.$$
\end{lem}

\begin{proof}
The result follows by \cref{eqn:greendet} and the definition of singularity.

\end{proof}

Let $R_n, R^{'}_n, Q_n, Q^{'}_n,$ and $\tilde{Q}_n$ be the random variables from \Cref{thm:boxev2} and for $l \in \mathbb{Z}$ set 

\begin{equation*}
F_{l,n,\varepsilon_0}^{1}=\{\omega: \max\{m_{L}(B^{-}_{[l+R_{n},l+R^{'}_{n}],\varepsilon_0,\omega}),m_{L}(B^{-}_{[l+Q_{n},l+Q^{'}_{n}],\varepsilon_0,\omega})\} \leq e^{-(\eta_0-\delta_0)(2n+1)} \}.
\end{equation*}

\begin{lem} \label{thm:evmeas}
There is an $N$ such that for $n> N$ and any $l \in \mathbb{Z}$, 
$$\mathbb{P}[F_{l,n,\varepsilon_{0}}^{1}] \geq 1 - 2m_{L}(\tilde{I})e^{-\delta_{0}(2n+1)}.$$
\end{lem}

\proof

With $0<\varepsilon_0 < 8\nu$ as above, choose $N$ such that the conclusion of \Cref{thm:boxev2} holds. Then for $n > N$,
\begin{align*}
m_{L}\times \mathbb{P}(B^{-}_{[l+R_{n},l+R^{'}_{n}],\varepsilon_0})&=\mathbb{E}(m_{L}(B^{-}_{[l+R_{n},l+R^{'}_{n}],\varepsilon_0,\omega}))\\
&=\int_{\mathbb{R}}\mathbb{P}(B^{-}_{[l+R_{n},l+R^{'}_{n}],\varepsilon_0,E})\; dm_{L}(E)\\
&\leq m_{L}(\tilde{I})e^{-\eta_0(2n+1)}.\\
\end{align*}

Applying the same reasoning to$B^{-}_{[l+Q_{n},l+Q^{'}_{n}],\varepsilon_0}$, by the estimate above and Chebyshev's inequality, 
$$e^{-(\eta_0-\delta_0)(2n+1)}\mathbb{P}[(F^{1}_{l,n,\varepsilon_0})^{c}]\leq 2m_{L}(\tilde{I})e ^{-\eta_0(2n+1)}.$$

The result follows by multiplying both sides of the last inequality by $e^{(\eta_{0} - \delta_{0})(2n+1)}$. 

\qed

\begin{rmk}
\Cref{thm:martini} is proved in \cite{xj} and used there to give a uniform (and quantitative) Craig-Simon estimate similar to the one in \cite{GeZhao}. 
\end{rmk}

\begin{lem} \label{thm:martini}
Let $Q(x)$ be a polynomial of degree $n-1$. Let $x_i=\cos\frac{2\pi(i+\theta)}{n}$, for $0< \theta <\frac{1}{2}, i=1,2,\ldots, n$. If $Q(x_i)\leq a^n$, for all $i$, then $Q(x)\leq Cna^n$, for all $x\in[-1,1]$, where $C=C(\theta)$ is a constant. 
\end{lem}

Set 
\begin{equation*}
F_{[a,b],\varepsilon}=\{\omega: |P_{[a,b],E,\omega}| \leq e^{(\gamma(E)+4\varepsilon)(b-a+1)} \text{ for all } E \in \tilde{I} \}.
\end{equation*}

\begin{lem} \label{thm:uniformcs}

There are $C >0$ and $N$ such that for $b-a+1 > N$, $$\mathbb{P}[F_{[a,b],\varepsilon}] \geq 1 - C(b-a+2)e^{-\eta_{\frac{\varepsilon}{4}}(b-a+1)}.$$

\end{lem}

\proof 

Since $\tilde{I}$ is compact and $\tilde{I} \cap D = \emptyset$, $\tilde{I}$ is contained in the union of finitely many compact intervals which all intersect $D$ trivially. Hence, it suffices to prove the result for all $E$ in one of these intervals. So fix one of these intervals, call it $I_1$. By continuity of $\gamma$ and compactness of $I_1$, if $\varepsilon > 0$, there is $\delta >0$ such that if $E$,$E^{'} \in I_1$ with $|E-E^{'}| < \delta$, then $|\gamma(E) - \gamma(E^{'})| < \frac{1}{4}\varepsilon$. Divide $I_1$ into sub-intervals of size $\delta$, denoted by $J_{k}=[E_{k}^{n},E_{k+1}^{n}]$ where $k=1,2,...C$. Additionally, let $E_{k,i}^{n} = E_{k}^{n} + (x_{i} + 1)\frac{\delta}{2}$.

By \Cref{thm:lrgdev}, there is an $N$ such that for $b-a+1>N$, 

$$\mathbb{P}[\{\omega : |P_{[a,b],E_{k,i}^{n},\omega}| \geq e^{(\gamma(E_{k,i}^{n}) + \frac{1}{4}\varepsilon)(b-a+1)}\}] \leq e^{-\eta_{\frac{\varepsilon}{4}}(b-a+1)}.$$

Put $\displaystyle F_{[a,b],k,\varepsilon} = \bigcup_{i=1}^{n} \{\omega : |P_{[a,b],E_{k,i}^{n},\omega}| \geq e^{(\gamma(E_{k,i}^{n}) + \frac{1}{4}\varepsilon)(b-a+1)}\}$ and $\gamma_{k} = \inf_{E \in J_k} \gamma(E)$. For $\omega \notin F_{[a,b],k,\varepsilon}$, $|P_{[a,b],E_{k,i}^{n},\omega}| \leq e^{(\gamma_{k}+ \frac{1}{2}\varepsilon)(b-a+1)}$. Thus, an application of \Cref{thm:martini} yields, for any such $\omega$,
$|P_{[a,b],E,\omega}| \leq e^{(\gamma(E) + \frac{3}{4}\varepsilon)(b-a+1)}$.

We have $$\mathbb{P}\left[\bigcup_{k=1}^{C} F_{[a,b],k,\varepsilon}\right] \geq 1 - C(b-a+2)e^{\eta_{\frac{\varepsilon}{4}}(b-a+2)}.$$

Thus, since $$\bigcup_{k=1}^{C} F_{[a,b],k,\varepsilon} \subset F_{[a,b],\varepsilon},$$ the result follows.

\qed

Let $n_{r,\omega}$ denote the center of localization (if it exists) for $\psi_{\omega,E}$ (i.e. $|\psi_{\omega,E}(n)| \leq |\psi_{\omega,E}(n_{r,\omega})|$). Note that by the results in \cite{gordon}, $n_{r,\omega}$ can be chosen as a measurable function of $\omega$. 

%Moreover, let $l \in \mathbb{Z}$ and let $R_n$, $Q_n$, and $Q^{'}_{n}$ be the scales from \Cref{thm:boxev2}.

Put $F^{4}_{l,n,\varepsilon} = \left((F_{[l+R_{n},l-1],\varepsilon} \cap F_{[l+1,l+R^{'}_{n}],\varepsilon}\right) \cap \left(\bigcup_{0 \leq k \leq 2m} (F_{[l+\tilde{Q}_{n}+k+1,l+Q^{'}_{n} ],\varepsilon} \cap F_{[l+Q_{n},l+\tilde{Q}_{n} + k - 1],\varepsilon}\right)$

and

\begin{equation*}
J_{l,n,\varepsilon} = F_{l,n,\varepsilon_0}^{1} \cap F_{l,n,\varepsilon}^{2} \cap F_{l,n,\varepsilon}^{3} \cap F^{4}_{l,n,\varepsilon}.
\end{equation*}

\begin{lem} \label{thm:lowerboundnumerator}
There is $N$ such that if $n>N$, $\omega \in J_{l,n,\varepsilon}$, with a generalized eigenfunction $\psi_{\omega,E}$ satisfying either
\begin{enumerate} 
\item $n_{r,\omega}=l$, or
\item $|\psi_{\omega}(l)| \geq \frac{1}{2},$ 
\end{enumerate}

then if $l+\tilde{Q}_{n}+k$ is $(\gamma(E)-8\varepsilon_0,\tilde{Q}_{n}+k - Q_{n},Q^{'}_{n}-\tilde
{Q}_{n}-k,E, \omega)$-singular (with $0 \leq k \leq 2m$), there exist
 $$l+R_{n}\leq y_1 \leq y_2 \leq l+R^{'}_{n}$$ and $E_{j}=E_{j,[l+Q_{n} + k,l+Q^{'}_{n}],\omega}$ such that 
\begin{equation*}
|P_{[l+R_{n},y_1],E_j,\tilde{\omega}}P_{[y_2,l+R^{'}_{n}],E_j,\omega}|\geq \frac{1}{2m_{L}(\tilde{I})\sqrt{m(2n+1)}}e^{(\gamma(E_j)-\varepsilon)(R^{'}_{n}-R_{n}+1)+(\eta_0-\delta_0)(2n+1)}.
\end{equation*}
\end{lem}

\begin{rmk}
Note that $y_1$ and $y_2$ depend on $\omega$ and $l$ but we do not include this subscript for notational convenience. In particular, this is done when the other terms in expressions involving $y_1$ or $y_2$ have the correct subscript and indicate the appropriate dependence.
\end{rmk}
\proof

Firstly, if $|\psi_{\omega}(l)| \geq \frac{1}{2}$, we may choose $N_1$ such that $l$ is $(\gamma(E)-8\varepsilon_{0}, -R_{n},R^{'}_{n}, E, \omega)$-singular for $n>N_1$. In the case that $n_{r,\omega}=l$, then there is an $N_2$ such that $l$ is naturally, $(\nu-8\varepsilon_{0},-R_{n},R^{'}_{n}, E, \omega)$-singular for all $n>N_2$. Choose $N_3$ so that $e^{\frac{-\nu}{2}n} < \text{dist}(I,\tilde{I})$ for $n>N_3$ and finally choose $N$ to be larger than  $N_1$, $N_2$, $N_3$ and the $N$'s from \Cref{thm:evmeas}, \Cref{thm:boxev2}, and \Cref{thm:uniformcs}.

Suppose that for some $n>N$, $l+\tilde{Q}_{n}+k$ is $(\gamma(E)-8\varepsilon_0,\tilde{Q}_{n}+k - Q_{n},Q^{'}_{n}-\tilde
{Q}_{n}-k,E, \omega)$-singular. By \Cref{thm:singpts} and \Cref{thm:uniformcs}, $E\in B^{-}_{[l+Q_{n},l+Q^{'}_{n}],\varepsilon_0,\omega}$. Note that all  eigenvalues of $H_{[l+Q_{n},l+Q^{'}_{n}],\omega}$ belong to $B^{-}_{[l+Q_{n},l+Q^{'}_{n}],\varepsilon_0,\omega}$. Since $P_{[l+Q_{n},l+Q^{'}_{n}],\tilde{E},\omega}$ is a polynomial in $\tilde{E}$, it follows that $B^{-}_{[l+Q_{n},l+Q^{'}_{n}],\omega,\varepsilon_{0}}$ is contained in the union of sufficiently small intervals centered at the eigenvalues of $H_{\omega,[l+Q_{n},l+Q^{'}_{n}]}$. Moreover, \Cref{thm:evmeas} gives $$m(B^{-}_{[l+Q_{n},l+Q^{'}_{n}],\omega,\varepsilon_{0}})\leq 2m_{L}(\tilde{I})e^{-(\eta_0-\delta_0)(2n+1)},$$ so we have the existence of $E_j=E_{j,[l+Q_{n},l+Q^{'}_{n}],\omega}$ so that $|E-E_j|\leq 2m_{L}(\tilde{I})e^{-(\eta_0-\delta_0)(2n+1)}$.

Applying the above argument with $l$ in place of $l+\tilde{Q}_{n} + k$ yields an eigenvalue $E_i=E_{i,\omega,[l+R_{n},l+R^{'}_{n}]}$ such that $E_i\in B^{-}_{[l+R_{n},l+R^{'}_{n}],\varepsilon_0,\omega}$ and $|E-E_i|\leq 2m_{L}(\tilde{I})e^{-(\eta_0-\delta_0)(2n+1)}$. Hence, $|E_i-E_j|\leq 4m_{L}(\tilde{I})e^{-(\eta_0-\delta_0)(2n+1)}$. By the previous line and the fact that $E_j\notin B^{-}_{[l+R_{n},l+R^{'}_{n}],\varepsilon,\omega}$, we see that $||G_{[l+R_{n},l+R^{'}_{n}],E_j,\omega}||\geq \frac{1}{4m_{L}(\tilde{I})}e^{(\eta_0-\delta_0)(2n+1)}$ so that for some $y_1,y_2$ with $l+R_{n}\leq y_1 \leq y_2 \leq l+R^{'}_{n}$, 
$$|G_{[l+R_{n},l+R^{'}_{n}],E_j,\omega}(y_1,y_2)|\geq \frac{1}{4m_{L}(\tilde{I})\sqrt{m(2n+1)}} e^{(\eta_0-\delta_0)(2n+1)}. $$

Additionally, another application of \Cref{thm:boxev2} yields, $|P_{[l+R_{n},l+R^{'}_{n}],E_j,\omega}|\geq e^{(\gamma(E_j)-\varepsilon)(R^{'}_{n}-R_{n}+1)}$. 

Thus, by \cref{eqn:greendet} we obtain
$$|P_{[l+R_{n},y_1],E_j,\omega}P_{[y_2,l+R^{'}_{n}],E_j,\omega}|\geq \frac{1}{4m_{L}(\tilde{I})\sqrt{m(2n+1)}}e^{(\gamma(E_j)-\varepsilon)(R^{'}_{n}-R_{n}+1)+(\eta_0-\delta_0)(2n+1)}. $$

\qed

\begin{lem} \label{thm:goodsetmeas}
There is a $\tilde{\eta}>0$ and $N$ such that $n>N$ implies $\mathbb{P}(J_{l,n,\varepsilon})\geq 1-e^{-\tilde{\eta}n}.$
\end{lem}

\proof

Let $\mathcal{A}_{1} = [2n+1 - m,(2n+3)m]$,
$\mathcal{A}_{2} = [3n+1 - m,(3n+3)m]$, and 
$\mathcal{A}_{3} = [n+1 - m,(n+3)m]$.
 %denote the possible values taken by $l + k - _{n}$ and $\mathcal{A}_{2}$ denote the possible values taken by $l + k + 1$.  
%Using the bounds on $R_{n}$, $Q_n$, and $Q^{'}_{n}$ established in \Cref{thm:boxev2}, for a sufficiently large $N$, the number of elements in $\mathcal{A}_{1}$ and $\mathcal{A}_{2}$ are at most $3n\varepsilon$ and $4n\varepsilon$ respectively. Moreover, $\inf \mathcal{A}_{2} - \sup \mathcal{A}_{1} \geq \frac{n}{2}(\langle L \rangle - 5\varepsilon)$.

Thus, $$ \left(\bigcup_{j_{1} \in \mathcal{A}_{1},j_{2} \in \mathcal{A}_{2}, j_2-j_1\geq n/2}F_{[j_{1},j_{2}],\varepsilon}\right) \cup \left(\bigcup_{j_{3} \in \mathcal{A}_{3},j_{4} \in \mathcal{A}_{1}, j_4-j_3\geq n/2}F_{[j_{3},j_{4}],\varepsilon}\right)$$ is contained in $$\left(\bigcup_{0 \leq k \leq 2m} (F_{[l+\tilde{Q}_{n}+k+1,l+Q^{'}_{n} ],\varepsilon} \cap F_{[l+Q_{n},l+\tilde{Q}_{n} + k - 1],\varepsilon})\right).$$

We also have $$\mathbb{P}\left[\left(\bigcup_{j_{1} \in \mathcal{A}_{1},j_{2} \in \mathcal{A}_{2}, j_2-j_1\geq n/2}F_{[j_{1},j_{2}],\varepsilon}\right) \cup \left(\bigcup_{j_{3} \in \mathcal{A}_{3},j_{4} \in \mathcal{A}_{1}, j_4-j_3\geq n/2}F_{[j_{3},j_{4}],\varepsilon}\right)\right] \geq 1 - 2n^{2}e^{-\eta\frac{n}{2}}$$

for sufficiently large $n$. 

Note that the same reasoning provides a similar estimate on $\mathbb{P}[F_{[l+R_{n},l-1],\varepsilon} \cap F_{[l+1,l+R^{'}_{n}],\varepsilon}]$.

%with a similar estimate holding for $\displaystyle \bigcup_{j_{3} \in \mathcal{A}_{3},j_{4} \in \mathcal{A}_{1}}F_{[j_{1},j_{2}],\varepsilon}$.

%\bigcup_{Q_{n} \leq k < Q_{n+1}} (F_{[l+k-R_{n},l+k-1],\varepsilon}$ and $$\mathbb{P}\left[\bigcup_{j_{1} \in \mathcal{A}_{1},j_{2} \in \mathcal{A}_{2}}F_{[j_{1},j_{2}],\varepsilon}\right] \geq 1 - n^{2}e^{-\eta\frac{n}{2}(\langle L \rangle - 5\varepsilon)}.$$

Choose $N$ as in \Cref{thm:lowerboundnumerator}, and note that by \Cref{thm:evmeas}, \Cref{thm:boxev2}, \Cref{thm:uniformcs} and the argument above, for  $n>N$,

\begin{equation*}
\mathbb{P}[J_{l,n,\varepsilon}]\geq 1 - 2m_{L}(\tilde{I})e^{-\delta_{0}(2n+1)} - 2m^{4}(2n+3)^{3}e^{-\eta_{\varepsilon}(\frac{n}{K})} - 2m^{2}(2n+3)^{2}e^{-\eta_{\varepsilon}(2n+1)} - 4n^{2}e^{-\eta_{\frac{\varepsilon}{4}}\frac{n}{2}}.
\end{equation*}

We may choose $\tilde{\eta}$ sufficiently close to $0$ and increase $N$ such that for $n>N$ , we have $$2m_{L}(\tilde{I})e^{-\delta_{0}(2n+1)} + 2m^{4}(2n+3)^{3}e^{-\eta_{\varepsilon}(\frac{n}{K})} + 2m^{2}(2n+3)^{2}e^{-\eta_{\varepsilon}(2n+1)} + 4n^{2}e^{-\eta_{\frac{\varepsilon}{4}}\frac{n}{2}} \leq e^{-\tilde{\eta}n},$$

and the result follows.

\qed

\begin{lem} \label{thm:lrgnumerator}

There is $N$ such that for $n>N$, any $\omega \in J_{l,n,\varepsilon}$, any $y_1$, $y_2$ with $l+R_{n} \leq y_1 \leq y_2 \leq l+R^{'}_{n}$ and any $E_j = E_{j,[l+k + Q_{n},l+Q^{'}_{n}],\omega}$ (with $0 \leq k \leq 2m$), $$|P_{[l+R_{n},y_1],E_j,\omega}P_{[y_2,l+R^{'}_{n}],E_j,\omega}|\leq e^{(\gamma(E_j)+\varepsilon)(R^{'}_{n}-R_{n}+1)}.$$

\end{lem}

\proof
By choosing $N$ so that \Cref{thm:boxev2} holds for $n>N$, we are led to consider three cases:
\begin{enumerate}
\item $l+R_{n} +\frac{n}{K} \leq y_1 \leq y_2 \leq l+R^{'}_{n} - \frac{n}{K}$,
\item $l+R_{n} +\frac{n}{K} \leq y_1 \leq l+R_{n}$, while $ l+R^{'}_{n} - \frac{n}{K} \leq y_2 \leq l+R^{'}_{n}$, and finally,
\item $l+R_{n} \leq y_1 \leq l+R_{n} +\frac{n}{K}$ and  $l+R_{n} + \frac{n}{K} \leq y_2 \leq l+R^{'}_{n}$.

\end{enumerate}

In the first case, \Cref{thm:boxev2} immediately yields: 
$$|P_{[l+R_{n},y_1],E_j,\omega}P_{[y_2,l+R^{'}_{n}],E_j,\omega}|\leq e^{(\gamma(E_j)+\varepsilon)(R^{'}_{n}-R_{n}+1)}.$$

In the second case, we have $|P_{[y_2,l+R^{'}_{n}],E_j,\omega}| \leq \tilde{M}^{\frac{n}{K}}$, while \Cref{thm:boxev2} gives $$|P_{[l+R_{n},y_1],E_j,\omega}| \leq e^{(\gamma(E_j)+\varepsilon)(\frac{n}{K})}.$$ By our choice of $K$, $\tilde{M}^{\frac{1}{K}} \leq e^{\frac{\nu}{2}} \leq e^{(\gamma(E_j)+\varepsilon)}$, so we again obtain the desired result.

Finally, in the third case, $|P_{[l+R_{n},y_1],E_j,\omega}P_{[y_2,l+R^{'}_{n}],E_j,\omega}|\leq \tilde{M}^{\frac{2n}{K}} \leq e^{(\gamma(E_j)+\varepsilon)(R^{'}_{n}-R_{n}+1)}$ (again by our choice of $K$).
\qed

\section{Spectral Localization} \label{sec:specloc}

\begin{thm} \label{thm:goodsetreg}
There is $N$ such that if $n>N$, $0 \leq k \leq 2m$, and $ \omega \in J_{l,n,\varepsilon}$, with a generalized eigenfunction $\psi_{\omega,E}$ satisfying either
\begin{enumerate}
\item $n_{r,\omega}=l$, or
\item $|\psi_{\omega}(l)| \geq \frac{1}{2}$,
\end{enumerate}

then $l+\tilde{Q}_{n} + k$ is $(\gamma(E)-8\varepsilon_0,\tilde{Q}_{n}+k - Q_{n},Q^{'}_{n}-\tilde
{Q}_{n}-k,E, \omega)$-regular.
\end{thm}

\proof

Choose $N$ so that \Cref{thm:lowerboundnumerator} and \Cref{thm:lrgnumerator} hold and $$\frac{1}{4m_{L}(\tilde{I})\sqrt{m(2n+1)}}e^{(\gamma(E_j) - \varepsilon)(R^{'}_{n}-R_{n}+1)+ (\eta_{0} - \delta_{0})(2n+1)} > e^{(\gamma(E_j)+\varepsilon)(R^{'}_{n}-R_{n}+1)}$$ for $n>N$. This can be done since $\varepsilon < \frac{\eta_{0} - \delta_{0}}{3m}$.

For $n>N$, we obtain the conclusion of the theorem. For if $l+Q_{n}+k$ was not $(\gamma(E)-8\varepsilon_0,\tilde{Q}_{n}+k - Q_{n},Q^{'}_{n}-\tilde
{Q}_{n}-k,E, \omega)$-regular, then by \Cref{thm:lowerboundnumerator}  $$|P_{[l+R_{n},y_1],E_j,\omega}P_{[y_2,l+R^{'}_{n}],E_j,\omega}|\geq\frac{1}{4m_{L}(\tilde{I})\sqrt{m(2n+1)}} e^{(\gamma(E_j)-\varepsilon)(R^{'}_{n}-R_{n}+1)+(\eta_0-\delta_0)(2n+1)}. $$ 
On the other hand, by \Cref{thm:lrgnumerator}, we have
$$|P_{[l+R_{n},y_1],E_j,\omega}P_{[y_2,l+R^{'}_{n}],E_j,\omega}|\leq e^{(\gamma(E_j)+\varepsilon)(R^{'}_{n}-R_{n}+1)}.$$

Our choice of $N$ in the first line of the proof yields a contradiction and completes the argument.

\qed

We are now ready to give the proof of \Cref{thm:specloc}. Again, $R_n$, $R^{'}_{n}, Q_n, Q^{'}_{n},$ and $\tilde{Q}_n$ are the scales from \Cref{thm:boxev2}.

\proof By \Cref{thm:goodsetmeas}, $\mathbb{P}[J_{0,n,\varepsilon} \text{ eventually }]=1$. Thus, we obtain $\tilde{\Omega}$ with $\mathbb{P}[\tilde{\Omega}]=1$ and for $\omega \in \tilde{\Omega}$, there is $N(\omega)$ such that for $n>N(\omega)$, $\omega \in J_{0,n,\varepsilon}$.

Since the spectral measures are supported by the set of generalized eigenvalues (e.g. \cite{Schnols}), it suffices to show for all $\omega \in \tilde{\Omega}$, every generalized eigenfunction with generalized eigenvalue $E \in I$ is in fact an $\ell^{2}(\mathbb{Z})$ eigenfunction which decays exponentially.

Fix an $\omega$ in $\tilde{\Omega}$ and let $\psi = \psi_{\omega,E} $ be a generalized eigenfunction for $H_{\omega}$ with generalized eigenvalue $E$. By \Cref{eqn:greenef}, and the bounds established on $R_n, R^{'}_n, Q_n, Q^{'}_n,$ and $\tilde{Q}_n$ from \Cref{thm:boxev2}, it suffices to show that there is $N(\omega)$ such that for $n>N{(\omega)}$, if $0 \leq k < 2m$, then $\tilde{Q}_{n}+k$ is $(\gamma(E)-8\varepsilon_0,\tilde{Q}_{n}+k - Q_{n},Q^{'}_{n}-\tilde
{Q}_{n}-k,E, \omega)$-regular. We may assume $\psi(0) \neq 0$, and moreover, by rescaling $\psi$, $|\psi(0)| \geq \frac{1}{2}$. Choose $N$ so that for $n > N$, the conclusions of \Cref{thm:goodsetreg} hold. Additionally, we may choose $N(\omega)$ such that for $n> N(\omega)$, $\omega \in J_{0,n,\varepsilon}$. For $n > \max\{N,N(\omega)\}$, the hypotheses of \Cref{thm:goodsetreg} are met, and hence $\tilde{Q}_{n}+k$ is $(\gamma(E)-8\varepsilon_0,\tilde{Q}_{n}+k - Q_{n},Q^{'}_{n}-\tilde
{Q}_{n}-k,E, \omega)$-regular. 
\qed

\section{Exponential Dynamical Localization} \label{sec:EDL}

The strategy used in this section follows \cite{GeZhao} with appropriate modifications needed to deal with the fact that single-step transfer matrices were not used in the large deviation estimates. In particular, the randomness in the conclusion of \Cref{thm:goodsetreg} will need to be accounted for.

The following lemma was shown in \cite{krugerj} and we state a version below suitable for obtaining EDL on the interval $I$.

Let $u_{k,\omega}$ denote an orthonormal basis of eigenvectors for $Ran(P_{I}(H_\omega))$, the range of the spectral projection of $H_{\omega}$ onto the interval $I$.

\begin{lem} \label{thm:kruger}\cite{krugerj}
Suppose there is $\tilde{C}>0$ and $\tilde{\gamma} > 0$ such that for any $s$, $l \in \mathbb{Z}$, 
$$ \mathbb{E}\left[\sum_{n_{r,\omega}=l}|u_{k,\omega}(s)|^{2}\right] \leq \tilde{C}e^{-\tilde{\gamma}|s-l|}.$$ Then there are $C>0$ and $\gamma>0$ such that for any $p,q \in \mathbb{Z}$, $$ \mathbb{E}[\sup_{t \in \mathbb{R}}|\langle \delta_{p},P_{I}(H_\omega)e^{itH_{\omega}}\delta_{q} \rangle|] \leq C(|p-q| + 1)e^{-\gamma|p-q|}.$$
\end{lem}

By \Cref{thm:kruger}, \Cref{thm:EDL} follows from \Cref{thm:krugerholds}.

\begin{thm} \label{thm:krugerholds}
There is $\tilde{C}>0$ and $\tilde{\gamma} > 0$ such that for any $s$,$l \in \mathbb{Z}$, $$ \mathbb{E}\left[\sum_{n_{r,\omega}=l}|u_{k,\omega}(s)|^{2}\right] \leq \tilde{C}e^{-\tilde{\gamma}|s-l|}.$$
\end{thm}

\proof 
We choose $N$ so that \Cref{thm:goodsetreg} and \Cref{thm:goodsetmeas} hold and for $0 \leq k \leq 2m$, set $\zeta_{j,k}=\min\{\tilde{Q}_j+k-Q_j,Q^{'}_j-\tilde{Q}_j-k\}$. We then choose $ 0 < c < (\nu - 8\varepsilon_{0})$ and increase $N$ such that if $j> N$ and $0 \leq k \leq 2m$, $c (\tilde{Q}_{j} + k) < (\nu - 8\varepsilon_{0}) 2\zeta_{j,k}$. This can be done using the bounds on $Q_n, Q^{'}_n,$ and $\tilde{Q}_n$ established in \Cref{thm:boxev2}. Finally, we choose $0 < \tilde{\eta}_{1} < \frac{\tilde{\eta}}{(\langle L \rangle + \varepsilon)3}$ and increase $N$ such that if $j > N$, $(\frac{\tilde{\eta}}{(\langle L \rangle + \varepsilon)3} - \tilde{\eta_{1}})j > \ln(j)$.

%Let $\mathcal{A}$ denote the set of $j \in \mathbb{Z}$ so that $l+\tilde{Q}_{j} \leq s < l+\tilde{Q}_{j+1}$ and note that $\mathcal{A}$ is finite. We also let $\zeta=\min\{\tilde{Q}_j+k-Q_j,Q^{'}_j-\tilde{Q}_j-k\}$.
%Using the bounds on $Q_n, Q^{'}_n,$ and $\tilde{Q}_n$ established in \Cref{thm:boxev2}, we can choose $ 0 < c < (\nu - 8\varepsilon_{0})$ and increase $N$ such that if $j> N$,

%for all $j \in \mathcal{A}$ and $\omega \in  \bigcup_{ j \in \mathcal{A}} J_{l,j,\varepsilon}$, $c (\tilde{Q}_{j} + k) < (\nu - 8\varepsilon_{0}) 2\zeta$.

Now consider $s$ and $l$ in $\mathbb{Z}$.

There are two cases to consider:
\begin{enumerate}
\item $s-l > 2m(N+1)$,
\item $s-l \leq 2m(N+1)$.
\end{enumerate}

Now suppose $n_{r,\omega} = l$ and $l+\tilde{Q}_{j} \leq s < l+\tilde{Q}_{j+1}$, and $\omega \in J_{l,j,\varepsilon}$, since $j > N$, using \Cref{thm:goodsetreg} and \cref{eqn:greenef},

\begin{align*} 
\begin{split}
|u_{r,\omega}(s)| &\leq 2|u_{r,\omega}(l)|e^{-(\gamma(E_{r,\omega})-8\varepsilon_{0}) \zeta} \\
&\leq 2|u_{r,\omega}(l)|e^{-(\nu-8\varepsilon_{0}) \zeta}.
\end{split}
\end{align*}

By orthonormality and H{\"o}lder's inequality,
\begin{align*}
\begin{split}
\sum_{n_{r,\omega} = l} |u_{r,\omega}(s)|^{2} &\leq 4\sum_{n_{r,\omega} = l} |u_{r,\omega}(l)|^{2}e^{-(\nu-8\varepsilon_{0}) 2\zeta_{j,k}}\\
&\leq 4\sum_{n_{r,\omega} = l} e^{-(\nu-8\varepsilon_{0}) 2\zeta_{j,k}}.
\end{split}
\end{align*}

We need to replace the randomness in the exponent above with an estimate that depends only on the point $s$.

Since $s = l + \tilde{Q}_{j} + k$ with $0 \leq k \leq 2m$,  by our choice of $c$ and $N$,

\begin{align*}
\sum_{n_{r,\omega} = l} e^{-(\nu-8\varepsilon_{0}) 2\zeta_{j,k}}
&\leq \sum_{n_{r,\omega} = l} e^{-c |s-l|}.
\end{align*}

Now let $\mathcal{A}$ denote the set of $j \in \mathbb{Z}$ so that $l+\tilde{Q}_{j} \leq s < l+\tilde{Q}_{j+1}$.

Finally, letting $J = \bigcup_{ j \in \mathcal{A}} J_{l,j,\varepsilon}$,  using the estimate provided by \Cref{thm:goodsetmeas} on $\mathbb{P}[J_{l,j,\varepsilon}]$ and our choice of $\tilde{\eta_{1}}$,
\begin{align*} 
\begin{split}
\mathbb{E}\left[\sum_{n_{r,\omega}=l} |u_{k,\omega}|^{2}\right]
&= \mathbb{E}\left[\sum_{n_{r,\omega}=l} |u_{r,\omega}|^{2} \chi_{J} + \sum_{n_{r,\omega}=l} |u_{r,\omega}|^{2} \chi_{J^{c}}\right]\\
& \leq Ce^{-c|s-l|} + Ce^{-\tilde{\eta_{1}}|s-l|}.
\end{split}
\end{align*}

In the second case, again by orthonormality, $\displaystyle \mathbb{E}\left[\sum_{n_{r,\omega}=l} |u_{r,\omega}(s)|^{2}\right] \leq 1$.

By letting $\tilde{\gamma} = \min\{c,\tilde{\eta_{1}}\}$ and choosing a sufficiently large $\tilde{C}> 0$, we obtain: $$ \mathbb{E}\left[\sum_{n_{r,\omega}=l}|u_{r,\omega}(s)|^{2}\right] \leq \tilde{C}e^{-\tilde{\gamma}|s-l|}.$$
\qed

\bibliographystyle{plain}
\bibliography{randomwordbib}

\end{document}